\documentclass[]{actacybpress}

\usepackage{cleveref}

\usepackage{amsmath, amssymb, amsbsy, amsthm}
\usepackage[utf8]{inputenc}
\usepackage{t1enc}

\usepackage[mathcal]{eucal}
\usepackage{array}
\def\contradict{\mathrel=\kern-2pt\mathrel|}

\usepackage{listings}
\usepackage[T1]{fontenc}

\usepackage{graphicx, algorithm, algorithmic}  
\usepackage{todonotes}
\usepackage{ebproof} 
\usepackage{mathtools} 
\usepackage[numbers,sort&compress]{natbib} 
\usepackage{url}

\newcommand{\fun}[4]{\mathtt{fun\ } #1(#2, \dots, #3) \rightarrow #4}
\newcommand{\funz}[2]{\mathtt{fun\ } #1() \rightarrow #2}
\newcommand{\case}[4]{\mathtt{case\ } #1 \mathtt{\ of\ } #2 \mathtt{\ then\ } #3 \mathtt{\ else\ } #4}
\newcommand{\letrec}[5]{\mathtt{letrec\ } #1(#2, \dots, #3) \rightarrow #4 \mathtt{\ in\ } #5}

\newcommand{\elet}[3]{\mathtt{let\ } #1 = #2 \mathtt{\ in\ } #3}
\newcommand{\apply}[3]{\mathtt{apply\ } #1(#2, \dots, #3)}
\newcommand{\applyz}[1]{\mathtt{apply\ } #1()}
\newcommand{\bif}[3]{\mathtt{call\ } #1(#2, \dots, #3)}

\newcommand{\idfs}{\mathcal{I}d}

\newcommand{\cons}[2]{{\normalfont \texttt{[} #1 \texttt{|} #2 \texttt{]}}}
\newcommand{\nil}[0]{{\normalfont \texttt{[]}}}

\newcommand{\send}[3]{\textit{send}(#1,#2,#3)}
\newcommand{\receive}[1]{\textit{rec}(#1)}
\newcommand{\arrive}[3]{\textit{arr}(#1,#2,#3)}
\newcommand{\self}[1]{\textit{self}(#1)}
\newcommand{\spawn}[3]{\textit{spawn}(#1, #2, #3)}
\newcommand{\internal}{\tau}
\newcommand{\term}{\Downarrow}
\newcommand{\setflag}{\textit{flag}}
\newcommand{\msg}[1]{\textit{msg}(#1)}
\newcommand{\exit}[2]{\textit{exit}(#1,#2)}
\newcommand{\link}{\textit{link}}
\newcommand{\unlink}{\textit{unlink}}

\newcommand{\subst}[2]{#1[#2]}
\newcommand{\pstep}[3]{#1 \xrightarrow[]{#2} #3}
\newcommand{\nstep}[4]{#1 \xrightarrow[]{#2 : #3} #4}
\newcommand{\nstepbr}[4]{#1 \xrightarrow[]{#2 : #3} \\& #4}
\newcommand{\nstepk}[3]{#1 \xrightarrow[]{#2}\!^* #3}
\newcommand{\nstepkbr}[3]{#1 \xrightarrow[]{#2}\!^* \\& #3}
\newcommand{\nstepstar}[2]{#1 \longrightarrow^* #2}

\newcommand{\lstep}[4]{\langle #1, #2 \rangle \rightarrow \langle #3, #4 \rangle}
\newcommand\doubleplus{+\kern-1.3ex+\kern0.8ex}
\newcommand{\concat}{\doubleplus}
\newcommand{\remfst}[3]{\textit{remFirst}(#1, #2, #3)}
\newcommand{\pop}[2]{\textit{rem}_1(#1, #2)}
\newcommand{\ins}[4]{#1[(#2, #3) \stackrel{\!+}{\mapsto} #4]} 

\newcommand{\rem}[2]{\textit{rem}(#1, #2)}
\newcommand{\cvt}[1]{\textit{convert}(#1)}
\newcommand{\cvtl}[1]{\textit{convert\_list}(#1)}
\newcommand{\mytt}[1]{{\normalfont \texttt{#1}}}

\newcommand{\idstack}{\mathcal{I}d}

\newcommand{\rewrites}[4]{\langle #1, #2 \rangle \longrightarrow \langle #3 , #4 \rangle}
\newcommand{\rewritesbr}[4]{\langle #1, #2 \rangle \longrightarrow \\ &\qquad\langle #3 , #4 \rangle}
\newcommand{\rewritesbrnoi}[4]{\langle #1, #2 \rangle \longrightarrow \\& \langle #3 , #4 \rangle}

\newcommand{\rewritesstarbr}[4]{\langle #1, #2 \rangle \longrightarrow^* \\& \langle #3, #4 \rangle}

\begin{document}



\title{A Formalisation of Core Erlang, a Concurrent Actor Language\thanks{Accepted for publication to Acta Cybernetica on 10th of May, 2023.}}
\author{Péter Bereczky\thanks{ELTE, Eötvös Loránd University, \email{berpeti@inf.elte.hu, daniel-h@elte.hu}, \orcid{0000-0003-3183-0712, 0000-0003-0261-0091}}, Dániel Horpácsi\thanksmark{2}, Simon Thompson\thanksmark{2}  \thanks{University of Kent, \email{S.J.Thompson@kent.ac.uk}, \orcid{0000-0002-2350-301X}}}


\maketitle

\begin{abstract}
In order to reason about the behaviour of programs described in a programming language, a mathematically rigorous definition of that language is needed. 
In this paper, we present a machine-checked formalisation of concurrent Core Erlang (a subset of Erlang) based on our previous formalisations of its sequential sublanguage. 
We define a modular, frame stack semantics, show how program evaluation is carried out with it, and prove a number of properties (e.g. determinism, confluence). Finally, we define program equivalence based on bisimulations and prove that side-effect-free evaluation is a bisimulation.
This research is part of a wider project that aims to verify refactorings to prove that particular program code transformations preserve program behaviour.

\keywords{Formal semantics, Formal verification, Concurrency, Actor model, Program equivalence, Bisimulation, Erlang, Core Erlang, Coq}
\end{abstract}

\section{Introduction}\label{sec:intro}

Our work here contributes to a wider project~\cite{harp} to reason about the correctness of refactorings for functional languages in general, and for Erlang~\cite{programmingerlang} in particular. In our terminology, refactoring is a code transformation that preserves the observable behaviour of programs. Our understanding of the state-of-the-art refactoring tools scene suggests that behaviour preservation (i.e. correctness) is subject to extensive testing, but formal verification is not yet used in practice. We aim to change this, at least in the case of Erlang, and develop higher assurance for refactorings by developing formal, machine-checked theories for program semantics, equivalence and program transformation.

Erlang is a dynamically-typed, impure, functional programming language, which excels at concurrency. Core Erlang~\cite{carlsson2000core} is a standard subset of Erlang that contains all the essential elements of Erlang, so that a semantics of Core Erlang can be extended to a semantics for the full language in a straightforward way.
In earlier work
we defined and implemented formal semantics for the sequential parts of Erlang and Core Erlang, including a reduction semantics for a subset of Erlang using the $\mathbb{K}$ framework~\cite{kerl}, and a natural semantics for a subset of Core Erlang, implemented in Coq~\cite{bereczky2020core,bereczky2020machine}. We have also implemented a functional big-step~\cite{owens2016functional} semantics for this subset of Core Erlang, and shown~\cite{coreerlang} that this semantics is equivalent to the natural semantics. In turn, the semantics was validated~\cite{bereczky2021validation} against the reference implementation of Erlang, namely the Erlang/OTP compiler~\cite{erlangotp}.

Having these semantics defined, we focused on proving the equivalence of programs. On the one hand, we are interested in using the semantics to prove particular pairs of programs equivalent, and on the other, the correctness of many local refactoring steps can be reduced to the equivalence of simple expressions. When developing precise, standard definitions of equivalence, we decided to bring our results to smaller-step semantics and developed a frame stack semantics and equivalence definitions~\cite{wand2018contextual} built on that for sequential Core Erlang~\cite{horpacsi2022program}. The frame stack style for semantics is beneficial for two reasons: it is well-suited to express various standard equivalence definitions~\cite{pitts1998operational}, and furthermore, the semantics of concurrent expressions can be defined more easily in small-step 
approaches~\cite{mosses2006formal}.

Our formalisations of (Core) Erlang are not the first ones. There are a number of other semantics for both sequential and concurrent subsets of (Core) Erlang on which our work has been based. The novelty of our work presented here lies in the fact that it remains more faithful to the language specification~\cite{carlsson2000core} and the reference manual~\cite{erlangref} than the others; for instance, unlike other works, we formalised exit signals and the signal ordering guarantee closely following the specification. We give a more detailed comparison and an overview of the related research in \Cref{sec:work}. Also worth pointing out is that our formal development is accompanied by a machine-checked implementation~\cite{coreerlangmini}.

We continue our formalisation efforts and in this paper we add concurrency to our frame stack semantics for Core Erlang.
In particular, we create the definition in a modular way: the sequential and process-local parts of the semantics can be replaced by a more complete formalised part of Core Erlang or Erlang (or indeed another programming language) without the need to rewrite the whole semantics.
The main contributions of this paper are the following:

\begin{itemize}
\item A modular, frame stack semantics for a concurrent subset of Core Erlang;
\item Proofs about the properties of the concurrent semantics;
\item Results on Core Erlang program equivalence verification using bisimulation.
\end{itemize}

The rest of the paper is structured as follows. In \Cref{sec:background} we introduce (Core) Erlang and our previous work informally, and define the syntax of the formalised sublanguage. In \Cref{sec:dynamic} we describe a modular, dynamic semantics of Core Erlang, focusing on the concurrent sublanguage, then in \Cref{sec:validation} we show the evaluation of simple concurrent programs and prove properties of the semantics. \Cref{sec:bisim} defines the concepts of program equivalence and the corresponding results, and then we discuss related work in \Cref{sec:work}. Finally, \Cref{sec:conclusion} discusses future work and concludes.

\section{Background}\label{sec:background}

As mentioned before, Erlang is a dynamically-typed, impure functional programming language.
The biggest strength of Erlang is that it really excels at concurrent computation, based on the actor model~\cite{agha1988Concurrent}. For this reason, Erlang was initially used in telecommunication and banking systems, but it now plays a role in high-availability, scalable web-based systems.

\subsection{The Erlang Model of Concurrency}\label{sec:model}

Erlang implements and extends the actor model~\cite{agha1988Concurrent}. An Erlang system contains lightweight processes (actors) that can spawn other processes to execute a particular task. Each process executes in its own space, and so they do not share memory. Processes can only communicate by asynchronous message passing. Each process has a message queue (mailbox), where incoming messages are stored in the order of their arrival. A process can select which messages to handle from its mailbox: messages do not need to be handled in the order in which they are received.

Besides messages, processes can also send and receive other signals~\cite{erlangrefprocs}, such as \emph{link}, \emph{unlink} and \emph{exit}. These additional signals can trigger potential changes in the state of the process immediately upon their arrival without being placed into the mailbox. The \emph{link} and \emph{unlink} signals create and remove, respectively, a bi-directional link between two processes, which represents a mutual dependency, and affects the handling of \emph{exit} signals. In general, \emph{exit} signals are used to indicate and initiate termination; they include a reason (describing why they were sent), and a flag indicating whether they were sent through a link (we call this value the \emph{link flag} of the \emph{exit} signal). If one of a pair of linked processes terminates, it will send an \emph{exit} signal to the other process via the link.
Processes can terminate for a number of reasons: having finished evaluation, receiving a particular \emph{exit} signal, or terminating abnormally (e.g. with an exception).

Processes have a flag called \texttt{'trap\_exit'} which, when set, causes \emph{exit} signals to be converted into messages (except in very particular circumstances), i.e.\ the process \emph{traps exits}. Based on this flag, the reason of the exit signal, and whether the \emph{exit} signal was sent through a link, there are three different outcomes (see~\cite{erlangrefprocs} and \Cref{sec:local}): the receiver process 1) terminates, 2) drops the \emph{exit} signal, or 3) converts the \emph{exit} signal to a message and adds it at the end of its mailbox.

In the next section, we present the syntax of the language under formalisation, which is a sublanguage of Core Erlang. Note that Core Erlang is not merely a standard subset of Erlang, it is also used in the compilation process as an intermediate step, and numerous programming languages based on the BEAM platform can be compiled to Core Erlang~\cite{elixircomp}. Furthermore, as for concurrency, the two languages implement essentially the same model.
For more details, we refer to the Erlang Programming book~\cite{programmingerlang} and the reference manual~\cite{erlangref}.

\subsection{Language Syntax}\label{sec:syntax}

In this section we discuss and extend the formal syntax of the sequential sublanguage of Core Erlang as presented in our previous work~\cite{horpacsi2022program}. For better readability, we use a syntax definition that abstracts over the concrete syntax of the language; however, any expressions written in this syntax can be simply transformed to Core Erlang.

\begin{definition}[Language syntax]
\begin{flalign*}
  v \in \textit{Val } ::=&\ i \mid a \mid \iota \mid \nil \mid \cons{v_1}{v_2} \mid \fun{f/k}{x_1}{x_k}{e} \\
  p \in \textit{Pat } ::=&\ i \mid a \mid \iota \mid \nil \mid \cons{p_1}{p_2} \mid x \\
  e \in \textit{Exp } ::=&\ v \mid x \mid f/k\mid \apply{e}{e_1}{e_k} \mid \case{e}{p}{e_1}{e_2} \\
  &\mid \elet{x}{e_1}{e_2} \mid \cons{e_1}{e_2} \mid \letrec{f/k}{x_1}{x_k}{e_0}{e_1} \\
  &\mid \bif{e}{e_1}{e_k} \mid \mathtt{receive}\ p_1 \rightarrow e_1; \dots p_k \rightarrow e_k\ \mathtt{end} \\
\end{flalign*}
\end{definition}

We use $i, k, n$ to range over integers, $a, f$ over atoms,
$x$ over variables, and $\iota$ over process identifiers.
$f/k$ denotes a function identifier, where $f$ is the function name, and $k$ is its arity.
The primitive values of the language are integers (denoted by numbers), atoms (strings of characters, enclosed in single quotation marks), and process identifiers (for simplicity, also denoted by numbers). Besides these, lists and functions are also values, and patterns are built from variables, integers, atoms and process identifiers, and formed into composite patterns as lists\footnote{Tuples are not included in this language, but would be handled similarly to parameter lists.}.

Note that process identifiers are not patterns in Core Erlang, but with process identifiers as patterns, we can distinguish them from other values of the language; in Erlang, the \texttt{is\_pid} function can be used instead. This distinction is needed to maintain the proof of coincidence of sequential equivalence definitions described in~\cite{horpacsi2022program} (namely, the coincidence of behavioural and contextual equivalence).

For simplicity of formalisation, functions are always named to enable explicit recursive calls, but in this paper we omit function names for readability when there are no recursive calls in the body expression. For lists, we use the standard notations of Erlang, that is a list $\cons{e_1}{\cons{e_2}{\cons{\dots}{\cons{e_n}{\nil}} \dots}}$ will be denoted by $\mytt{[}e_1\mytt{,} e_2\mytt{,} \dots\mytt{,} e_n\mytt{]}$. Note that we also include Erlang's improper lists (such as $\cons{\texttt{1}}{\texttt{2}}$), but these do not require specific care in the semantics rules.

Expressions of the sequential sublanguage are values, variables, function identifiers, binding expressions (both \verb|let| and \verb|letrec|), function applications (\verb|apply|), pattern matching (\verb|case|) expressions\footnote{This expression is a simplified version of Core Erlang's \texttt{case}, restricting it to only two branches.}. 

We extend the syntax (as in~\cite{horpacsi2022program}) with two language elements in this work, the first one is BIF (built-in function) call (denoted by $\bif{e}{e_1}{e_k}$), the second is the \verb|receive| expression. BIFs are used to implement both sequential (e.g. addition of integers) and concurrent features of the language.

In particular, the concurrency model introduced in the previous subsection is implemented as follows:
\begin{itemize}
\item the \verb|'!'| BIF is used to send messages;
\item a \texttt{receive} expression is used to select a message from the process mailbox by means of pattern matching;
\item processes are created with the \texttt{'spawn'} BIF (taking a function and its parameters as arguments for the new process to evaluate);
\item \emph{link}, \emph{unlink} and \emph{exit} signals can be sent with the identically named BIFs;
\item the \texttt{'process\_flag'} BIF is used to set the \texttt{'trap\_exit'} flag.
\end{itemize}

The syntax we presented here is implemented in Coq using the nameless variable representation~\cite{deBruijn}. This way, we reuse existing approaches to define capture-avoiding, parallel substitutions~\cite{autosubst}. Nonetheless, we use named variables in this paper for readability.
Substitutions are denoted by $\subst{e}{x_1 \mapsto v_1, \dots, x_k \mapsto v_k}$, which results in replacing $x_1, \dots, x_k$ variables simultaneously with $v_1, \dots, v_k$ values in the expression $e$. We omit further details about substitutions and static semantics since they are not in the scope of this paper. For further details we refer to our previous work~\cite{horpacsi2022program} and to the formalisation~\cite{coreerlangmini}. Next, we show an example expression in the syntax presented above:

\begin{example}[A simple map function in Core Erlang]\label{example:cerlmap}
The following snippet shows a simple sequential Core Erlang function that transforms the elements of a list by applying the function \texttt{F} to each member.
Since it is a rather simple definition, we present it in concrete syntax for better 
readability. To evaluate the function, it suffices to substitute the body of the \verb|letrec| (denoted by \verb|...|) with an application of \verb|'mm'/2|.
\begin{lstlisting}[escapeinside={(*}{*)}]
letrec 'mm'/2 =
  fun(F, E) ->
    case E of [H|T]
      then [ apply F(H) | apply 'mm'/2(F, T) ]
      else []
    end
  in ...
\end{lstlisting}
\vspace{-1em}\hfill$\triangle$
\end{example}

The syntax of the sequential sublanguage is minimal, but representative.

\section{Dynamic Semantics}\label{sec:dynamic}

In this section we explain the dynamic semantics of the formalised Core Erlang subset. We present a three-layered, modular semantics for the language such that the sequential parts of the semantics can be replaced by a more complete formalised part of Core Erlang, Erlang, or another programming language entirely.

\begin{table}[h!]
\centering
\renewcommand{\arraystretch}{1.5}
\begin{tabular}{c|c|c}
    Layer name & Notation & Description \\
\hline
    Inter-process semantics (\Cref{sec:inter}) & $\xrightarrow{\iota:a}$ & System-level reductions \\
    Process-local semantics (\Cref{sec:local}) & $\xrightarrow{a}$ & Process-level reductions \\
    Sequential semantics (\Cref{sec:sequential}) & $\longrightarrow$ & Computational reductions
\end{tabular}
\caption{Layers of the semantics}
\end{table}

\subsection{Sequential Semantics}\label{sec:sequential}

First, we briefly present the sequential semantics~\cite{horpacsi2022program} on which we base the concurrent formalisation.
We highlight that the specification of Core Erlang~\cite{carlsson2000core} does not define the evaluation order of subexpressions, but the compiler employs a leftmost-innermost strategy~\cite{neuhausser2007abstraction}:
during the standard translation of Erlang, the evaluation order is enforced by nested \verb|let| expressions in Core Erlang.
Furthermore, lists in Core Erlang are evaluated from the right\footnote{The reference implementation generates a bytecode sequence that evaluates the list tail before evaluating the list head.}. The compiler's evaluation strategy \emph{is} reflected in our definition.

The semantics has been formally defined as a frame stack semantics~\cite{pitts1998operational}. This definition style resembles reduction semantics~\cite{felleisen1986control}, but the reduction context is deconstructed into a stack of evaluation frames with holes denoted by $\Box$. The frame stack can be regarded as the continuation of the computation.

\begin{definition}[Syntax of frames, frame stacks]

\begin{flalign*}
F \in \textit{Frame} &::= \bif{\Box}{e_1}{e_k} \mid \bif{v}{\Box}{e_k} \mid \cdots \mid \bif{v}{v_1}{\Box} \\
& \mid \apply{\Box}{e_1}{e_k} \mid \apply{v}{\Box}{e_k} \mid \cdots \mid \apply{v}{v_1}{\Box} \\
& \mid \elet{x}{\Box}{e_2} \mid \case{\Box}{p}{e_2}{e_3}\\
& \mid \cons{e_1}{\Box} \mid \cons{\Box}{v_2} \\
K \in \textit{FrameS}&\textit{tack} ::= \idfs \mid F :: K
\end{flalign*}

\end{definition}

\begin{figure}[tb!]

    \begin{flalign*}
    &\rewrites{K}{\elet{x}{e_1}{e_2}}{\elet{x}{\Box}{e_2} :: K}{e_1} \\
    &\rewrites{K}{\cons{e_1}{e_2}}{\cons{e_1}{\Box} :: K}{e_2} \\
    &\rewrites{K}{\apply{e}{e_1}{e_k}}{\apply{\Box}{e_1}{e_k} :: K}{e} \\ 
    &\rewrites{K}{\bif{e}{e_1}{e_k}}{\bif{\Box}{e_1}{e_k} :: K}{e} \\
    &\rewritesbr{K}{\letrec{f/k}{x_1}{x_k}{e_0}{e}}{K}{\subst{e}{f/k \mapsto \fun{ f/k}{x_1}{x_k}{e_0}}} \\
    &\rewrites{K}{\case{e_1}{p}{e_2}{e_3}}{\case{\Box}{p}{e_2}{e_3} :: K}{e_1} \\[1em]
    \hline\\
    &\rewrites{\apply{\Box}{e_1}{e_k} :: K}{v}{\apply{v}{\Box}{e_k} :: K}{e_1} \\ 
    &\rewrites{\bif{\Box}{e_1}{e_k} :: K}{v}{\bif{v}{\Box}{e_k} :: K}{e_1} \\
    &\rewritesbr{\mathtt{apply\ } v(v_1, \dots, v_{i-1}, \Box, e_{i + 1}, \dots, e_k) :: K}{v_i}{\mathtt{apply\ } v(v_1, \dots, v_{i-1}, v_i, \Box, e_{i + 2}, \dots, e_k) :: K}{e_{i+1}} \qquad(\text{if } i < k)\\ 
    &\rewritesbr{\mathtt{call\ } v(v_1, \dots, v_{i-1}, \Box, e_{i + 1}, \dots, e_k) :: K}{v_i}{v(v_1, \dots, v_{i-1}, v_i, \Box, e_{i + 2}, \dots, e_k) :: K}{e_{i+1}} \qquad(\text{if } i < k)\\
    &\rewrites{\cons{e_1}{\Box} :: K}{v_2}{\cons{\Box}{v_2} :: K}{e_1}\\[1em]
    \hline\\
    &\rewrites{\applyz{\Box} :: K}{\funz{f/0}{e}}{K}{\subst{e}{f/0 \mapsto \funz{f/0}{e}}} \\
    &\rewritesbr{\apply{(\fun{f/k}{x_1}{x_k}{e})}{v_1}{\Box} :: K}{v_k}{K}{\subst{e}{f/k \mapsto \fun{f/k}{x_1}{x_k}{e}, x_1 \mapsto v_1, \dots, x_k \mapsto v_k}} \\
    &\rewrites{\mathtt{call\ } \texttt{'+'}(i_1, \Box) :: K}{i_2}{K}{i_1 + i_2} \\
    &\rewrites{\elet{x}{\Box}{e_2}::K}{v}{K}{e_2[x \mapsto v]} \\
    &\rewrites{\cons{\Box}{v_2} :: K}{v_1}{K}{\cons{v_1}{v_2}}\\
    &\rewrites{\case{\Box}{p}{e_2}{e_3}::K}{v}{K}{\subst{e_2}{\textit{match}(p,v)}}\ \  (\text{if } \textit{is\_match}(p,v)) \\
    &\rewrites{\case{\Box}{p}{e_2}{e_3}::K}{v}{K}{e_3} \ \ \ \ \ \ \ \ \ \ \ \ \qquad(\text{if } \neg\textit{is\_match}(p,v)) \\
    \end{flalign*}
\caption{Sequential semantics of Core Erlang}
\label{fig:sequential}
\end{figure}

For the stacks, we use the following notations: $\idfs$ denotes the empty stack and $F :: K$ denotes adding frame $F$ to the top of stack $K$. Next, we introduce two metatheoretical functions for pattern matching:

\begin{itemize}
\item $\textit{is\_match}(p, v)$: determines whether the value $v$ matches the pattern $p$: that is they have been built up with the same constructs of Core Erlang up to pattern variables.
\item $\textit{match}(p, v)$: if the value $v$ matches the pattern $p$, this function returns a substitution which contains the result of the pattern matching in form of a mapping from pattern variables to values.
\end{itemize}

We present the sequential semantics rules in \Cref{fig:sequential}. We use $\rewrites{K}{e}{K'}{e'}$ to denote one reduction step between configurations consisting of a frame stack and an expression to be evaluated.
Recall that $v$, $v_i$ are used for values, $i$, $i_j$ for integer values, and $e$, $e_k$ for (unevaluated) expressions.

The biggest advantage of this semantics definition is that there are no premises in the reduction rules about the reduction of subexpressions since they have been put into the frame stack. Therefore the propagation of concurrent actions to this level is not necessary (i.e. there are no labels on the reduction rules). On the other hand, its disadvantage is that the complex syntax of frames is needed to be defined separately from the syntax of the language. 

The reduction rules can be categorised into three groups:

\begin{itemize}
\item Rules that extract the first redex from language constructs, and put the remainder with a hole into the frame stack.
\item Rules that modify the top frame of the stack by putting the calculated value into the hole, and obtaining the next reducible expression from the same frame.
\item Rules that remove the top element of the frame stack, which also marks that the subexpression has been completely reduced.
\end{itemize}

The evaluation of any language element (except \verb|letrec|) includes using exactly one rule once from the first and third categories. We note that this would change if exceptions and exception handler expressions were present in the sequential language. The connection between exceptions and signals is that when an exception terminates a process, it will emit an exit signal with the details of the exception. The presence of exceptions does not affect the modularity of the definition, but it would require consistent modifications in multiple layers.

\begin{example}[Sequential evaluation of \Cref{example:cerlmap}]\label{ex:seqeval}

We use $\longrightarrow^*$ to denote the reflexive, transitive closure of the relation $\longrightarrow$. For simplicity, we denote the successor function $\texttt{fun}(\texttt{X}) \rightarrow \texttt{call '+'/2}(\texttt{X}, \texttt{1})$ with $f$ in the following example. We also use $\textit{mm}$ to denote the function bound inside the \verb|letrec| expression in \Cref{example:cerlmap}.

The first step is to evaluate the head of the application $\textit{mm}$ to itself (since it is a function). Next, the parameter function $f$ is reduced to itself. Thereafter, the parameter list is reduced ($\mytt{[}\texttt{0}\mytt{,}\texttt{1}\mytt{,}\texttt{2}\mytt{]}$) by deconstructing it starting from the back, pushing the head elements of the sublists into the frame stack. Actually, the semantics just checks in this case that all of these elements are values, and then the list is reconstructed. These actions transform the type of the parameter list from $\cons{e_1}{e_2}$ to $\cons{v_1}{v_2}$, this is the reason why they are necessary, although, there are two seemingly identical configurations in the reduction sequence.

\begin{flalign*}
&\langle \idfs, \texttt{letrec\ } \texttt{'mm'/2} = mm \texttt{ in apply\ 'mm'/2}(f, \mytt{[}\texttt{0}\mytt{,}\texttt{1}\mytt{,}\texttt{2}\mytt{]}) \rangle \longrightarrow \\
&\rewritesbrnoi{\idfs}{\texttt{apply } \textit{mm}(f, \mytt{[}\texttt{0}\mytt{,}\texttt{1}\mytt{,}\texttt{2}\mytt{]})}
{\texttt{apply } \textit{mm}(\Box, \mytt{[}\texttt{0}\mytt{,}\texttt{1}\mytt{,}\texttt{2}\mytt{]})::\idfs}{f}  \longrightarrow \\
&\langle \texttt{apply } \textit{mm}(f, \Box)::\idfs, \mytt{[}\texttt{0}\mytt{,}\texttt{1}\mytt{,}\texttt{2}\mytt{]}\rangle \longrightarrow^*\\
&\langle \cons{\texttt{2}}{\Box}::\cons{\texttt{1}}{\Box}::\cons{\texttt{0}}{\Box}::\texttt{apply } \textit{mm}(f, \Box)::\idfs, \nil \rangle \longrightarrow^* \\
&\langle \texttt{apply } \textit{mm}(f, \Box)::\idfs, \mytt{[}\texttt{0}\mytt{,}\texttt{1}\mytt{,}\texttt{2}\mytt{]}\rangle
\end{flalign*}

Thereafter, the function \textit{mm} is applied by substituting the previous list into its body expression. The pattern in the \verb|case| expression matches the parameter list, thus the first clause will be evaluated.

\begin{flalign*}
&\rewritesbrnoi{\texttt{apply } \textit{mm}(f, \Box)::\idfs}{\mytt{[}\texttt{0}\mytt{,}\texttt{1}\mytt{,}\texttt{2}\mytt{]}}
{\idfs}{\case{\mytt{[}\texttt{0}\mytt{,}\texttt{1}\mytt{,}\texttt{2}\mytt{]}}{\\&\qquad\qquad\cons{\texttt{H}}{\texttt{T}}}{\cons{\texttt{apply } f(\texttt{H})}{\texttt{apply } \textit{mm}(f, \texttt{T})}}{\nil}} \longrightarrow \\
&\langle \idfs, \cons{\texttt{apply } f(\texttt{0})}{\texttt{apply } \textit{mm}(f, \mytt{[}\texttt{1}\mytt{,}\texttt{2}\mytt{]})} \rangle
\end{flalign*}

Next, we continue the evaluation with the tail of the list (since lists are evaluated backwards). Again we evaluate the application of \textit{mm} and reduce the \verb|case| expression, etc. The recursion stops when the list has been consumed. The last sublist, $\nil$ will not match the pattern of the \verb|case| expression, thus the application of \textit{mm} will leave $\nil$ unchanged while being removed from the stack. These reduction steps built up a sequence of applications inside the stack.

\begin{flalign*}
&\langle \idfs, \cons{\texttt{apply } f(\texttt{0})}{\texttt{apply } \textit{mm}(f, \mytt{[}\texttt{1}\mytt{,}\texttt{2}\mytt{]})} \rangle\longrightarrow^* \\
&\rewritesstarbr{\texttt{apply } \textit{mm}(f, \Box)::\cons{\texttt{apply } f(\texttt{0})}{\Box}::\idfs}{\mytt{[}\texttt{1}\mytt{,}\texttt{2}\mytt{]}}
{\texttt{apply } \textit{mm}(f, \Box)::\cons{\texttt{apply } f(\texttt{1})}{\Box}::\cons{\texttt{apply } f(\texttt{0})}{\Box}::\idfs}{\mytt{[}\texttt{2}\mytt{]}} \longrightarrow^* \\
&\rewritesstarbr{\texttt{apply } \textit{mm}(f, \Box)::\cons{\texttt{apply } f(\texttt{2})}{\Box}::\cons{\texttt{apply } f(\texttt{1})}{\Box}::
\\&\qquad\cons{\texttt{apply } f(\texttt{0})}{\Box}::\idfs}{\nil}
{\cons{\texttt{apply } f(\texttt{2})}{\Box}::\cons{\texttt{apply } f(\texttt{1})}{\Box}::\cons{\texttt{apply } f(\texttt{0})}{\Box}::\idfs}{\nil}
\end{flalign*}

Thereafter, the function applications can be evaluated for the elements of the list. First, the top element of the frame is extracted while $\nil$ is placed back. The application of $f$ increases \verb|2| to \verb|3|. Combining the top element of the frame ($\cons{\Box}{\nil}$) and \verb|3|, we obtain the list value $\mytt{[}\texttt{3}\mytt{]}$.

\begin{flalign*}
&\langle\cons{\texttt{apply } f(\texttt{2})}{\Box}::\cons{\texttt{apply } f(\texttt{1})}{\Box}::\cons{\texttt{apply } f(\texttt{0})}{\Box}::\idfs, \nil \rangle
\longrightarrow \\
&\langle \cons{\Box}{\nil}::\cons{\texttt{apply } f(\texttt{1})}{\Box}::\cons{\texttt{apply } f(\texttt{0})}{\Box}::\idfs, \texttt{apply } f(\texttt{2})\rangle \longrightarrow^* \\
&\langle \cons{\Box}{\nil}::\cons{\texttt{apply } f(\texttt{1})}{\Box}::\cons{\texttt{apply } f(\texttt{0})}{\Box}::\idfs, \texttt{3}\rangle \longrightarrow \\
&\langle \cons{\texttt{apply } f(\texttt{1})}{\Box}::\cons{\texttt{apply } f(\texttt{0})}{\Box}::\idfs, \mytt{[}\texttt{3}\mytt{]} \rangle
\end{flalign*}

For the other two elements, we omit the previous steps and just show how the list inside the frame stack is reconstructed.

\begin{flalign*}
&\rewritesstarbr{\cons{\texttt{apply } f(\texttt{1})}{\Box}::\cons{\texttt{apply } f(\texttt{0})}{\Box}::\idfs}{\mytt{[}\texttt{3}\mytt{]}}
{\cons{\texttt{apply } f(\texttt{0})}{\Box}::\idfs}{\mytt{[}\texttt{2}\mytt{,}\texttt{3}\mytt{]}} \longrightarrow^*
\langle \idfs , \mytt{[}\texttt{1}\mytt{,}\texttt{2}\mytt{,}\texttt{3}\mytt{]} \rangle
\end{flalign*}
\hfill$\triangle$
\end{example}

In the next section, we show how we built the concurrent semantics on top of the frame stack relation.

\subsection{Processes, Signals and  Actions}\label{sec:concurrentbackground}

In this section we formalise the notions of \emph{processes}, \emph{signals} and \emph{actions}, on which we build in the next two sections where we describe the concurrent semantics of Core Erlang, first the process-local semantics and then the inter-process semantics. In the remainder of this section we also establish some metatheoretical notation that we use in presenting the semantics.

\begin{definition}[Core Erlang processes]A process ($p \in \textit{Process}$) is either dead or alive.
\begin{itemize}
\item A live process is a quintuple $(K, e, q, \textit{pl},\textit{flag})$, where 
$K$ denotes a frame stack, $e$ is an expression, $q$ is the mailbox (represented as a list of values).   \textit{pl} is the set of linked processes (a list of process identifiers), and \textit{flag} is the status of the \texttt{'trap\_exit'} flag.
\item A terminated (or dead) process is a list of linked process identifiers.
\end{itemize}
\end{definition}

As described earlier, Erlang and Core Erlang implement the actor model~\cite{agha1988Concurrent} for asynchronous communication between processes by message passing. Besides messages, there are other signals that can be sent between the processes (we formalise \textit{exit}, \textit{link}, and \textit{unlink} signals beside messages) which potentially change the state of the process upon arrival without being put into the mailbox. 

Actions represent the effects that characterise concurrency: message send and arrival in a mailbox, processing a mailbox with \emph{receive}, process creation, and so on. An action will have an effect on individual processes (in the process-local semantics) and also \emph{between} processes in the system level, inter-process semantics.

We define the following signals and actions of the semantics.

\begin{definition}[Signals and Actions]\label{def:sigact}
\begin{flalign*}
  s \in \textit{Signal } &::= \msg{v} \mid \exit{v}{b} \mid \link \mid \unlink \\
  a \in \textit{Action } &::= \send{\iota_1}{\iota_2}{s} \mid \receive{v} \mid \self{\iota} \mid \arrive{\iota_1}{\iota_2}{s} \mid \spawn{\iota}{e_1}{e_2} \\ &\mid \internal \mid\ \term\ \mid \setflag
\end{flalign*}
\end{definition}

Signals can be messages (parametrised by a value), exits (parametrised by a reason value and a flag whether the exit was sent through a link), links, and unlinks (which do not have parameters). The source and destination process identifiers are handled by actions, thus they are not included in the signals.
We explain the syntax of actions as follows:
\begin{itemize}
\item Signal sending (\textit{send}) and signal arrival (\textit{arr}) actions carry a signal as a parameter, as well as the source and target process identifiers which are propagated from the inter-process semantics.
\item \textit{rec} actions have as parameter the message that is to be removed from the mailbox. There is no need to include process identifiers since these actions denote a process-local step and the removable message is already present in the mailbox of the process.
\item $\textit{self}$ actions contain the identifier of the executing process as a parameter, which was obtained from the inter-process semantics.
\item \textit{spawn} actions include the new process identifier (propagated from the inter-process semantics), a function expression, and its actual parameters (as a Core Erlang list). The spawned process will execute this function with the given parameters.
\item A sequential ($\internal$) action denotes one reduction step with the sequential semantics.
\item Termination ($\term$) actions denote either normal termination or the execution of the single-parameter \texttt{'exit'} BIF.
\item $\setflag$ actions denote the execution of the \texttt{'process\_flag'} BIF (which does not necessarily change the state of the \texttt{'trap\_exit'} flag).
\end{itemize}

Actions are used as the labels of the one-step evaluation relation. Next, we define the following \emph{metatheoretical} functions and notations for the next sections:

\begin{itemize}
\item $\textit{tt}$ denotes the metatheoretical true, while $\textit{ff}$ denotes false.
\item $x::\textit{xs}$ denotes a list with $x$ as the first element and $\textit{xs}$ as the tail.
\item $[]$ denotes the empty list.
\item $[x_1, \dots, x_n] = x_1 :: (x_2 :: (\dots x_n :: []) \dots)$.
\item $\pop{x}{l}$: creates a list by removing the first occurrence of $x$ from $l$.
\item $\rem{x}{l}$: creates a list by removing all occurrences of $x$ from $l$.
\item $\textit{map}(\textit{fn}, l)$: constructs a list by applying the metatheoretical function $\textit{fn}$ to the elements of $l$.
\item $l_1 \concat l_2$: constructs a list to represent the concatenation of $l_1$ and $l_2$.
\item $\cvt{b}$: maps $\textit{tt}$ to \texttt{'true'} and $\textit{ff}$ to \texttt{'false'}.
\item $\cvt{v}$: maps \texttt{'true'} to $\textit{Some tt}$ and \texttt{'false'} to $\textit{Some ff}$, for other inputs, it returns $\textit{None}$.
\end{itemize}

\subsection{Process-Local Semantics}\label{sec:local}

\begin{figure}[t!]
\centering

\begin{equation}
\begin{prooftree}
\hypo{\lstep{K}{e}{K'}{e'}}
\infer1{\pstep{(K,e,q,\textit{pl},b)}{\internal}{(K',e',q,\textit{pl},b)}}
\end{prooftree}
\label{OS:seq}
\tag{\textsc{Seq}}
\end{equation}

\begin{equation}
\begin{prooftree}
\hypo{\pstep{(K,e,q,\textit{pl},b)}{\arrive{\iota_1}{\iota_2}{\msg{v}}}{(K,e,q \concat [v],\textit{pl},b)}}
\end{prooftree}
\label{OS:msg}
\tag{\textsc{Msg}}
\end{equation}

\begin{equation}
\begin{prooftree}
\hypo{(\iota_1 \neq \iota_2 \land b = \textit{ff} \land v = \texttt{'normal'}) \lor
      (\iota_1 \notin \textit{pl} \land b_e = \textit{tt} \land \iota_1 \neq \iota_2)}
\infer1{\pstep{(K,e,q,\textit{pl},b)}{\arrive{\iota_1}{\iota_2}{\exit{v}{b_e}}}{(K,e,q,\textit{pl},b)}}
\end{prooftree}
\label{OS:exitdrop}
\tag{\textsc{ExitDrop}}
\end{equation}

\begin{equation}
\begin{prooftree}
\hypo{(v = \texttt{'kill'} \land b_e = \textit{ff} \land v' = \texttt{'killed'}) \lor}
\infer[no rule]1{(b = \textit{ff} \land v \neq \texttt{'normal'} \land v' = v \land (b_e = \textit{tt} \rightarrow \iota_1 \in \textit{pl}))
      \lor}
\infer[no rule]1{(b = \textit{ff} \land v = \texttt{'normal'} = v' \land \iota_1 = \iota_2)}
\infer1{\pstep{(K,e,q,\textit{pl},b)}{\arrive{\iota_1}{\iota_2}{\exit{v}{b_e}}}{\textit{map}\ (\lambda \iota \Rightarrow (\iota,v'))\ \textit{pl}}}
\end{prooftree}
\label{OS:exitterm}
\tag{\textsc{ExitTerm}}
\end{equation}

\begin{equation}
\begin{prooftree}
\hypo{b = \textit{tt} \land ((b_e = \textit{ff} \land v \neq \texttt{'kill'}) \lor (b_e = \textit{tt} \land \iota_1 \in \textit{pl}))}
\infer1{\pstep{(K,e,q,\textit{pl},b)}{\arrive{\iota_1}{\iota_2}{\exit{v}{b_e}}}{(K,e,q \concat [\texttt{[}\texttt{'EXIT'}\texttt{,} \iota_1\texttt{,} v\texttt{]}],\textit{pl},b)}}
\end{prooftree}
\label{OS:exitconv}
\tag{\textsc{ExitConv}}
\end{equation}

\begin{equation}
\begin{prooftree}
\hypo{\pstep{(K,e,q,\textit{pl},b)}{\arrive{\iota_1}{\iota_2}{\link}}{(K,e,q,\iota_1::\textit{pl},b)}}
\end{prooftree}
\label{OS:linkarr}
\tag{\textsc{LinkArr}}
\end{equation}

\begin{equation}
\begin{prooftree}
\hypo{\pstep{(K,e,q,\textit{pl},b)}{\arrive{\iota_1}{\iota_2}{\unlink}}{(K,e,q,\rem{\iota_1}{\textit{pl}},b)}}
\end{prooftree}
\label{OS:unlinkarr}
\tag{\textsc{UnlinkArr}}
\end{equation}

\caption{Process local semantics (part 1)}
\label{fig:localsem1}
\end{figure}

Next, we show the process-local semantics (see \Cref{fig:localsem1}, \Cref{fig:localsem2}, and \Cref{fig:localsem3}), denoted by $\pstep{p}{a}{p'}$,  which describes the one-step evaluation of actions by individual processes. We primarily built this semantics by following the techniques of Fredlund's formalisation~\cite{fredlund2001framework}, since it has the widest coverage of language features among previous semantics. We note that the evaluation of the parameters of BIF calls $\bif{e}{e_1}{e_k}$ is handled by the sequential semantics (see~\Cref{sec:sequential}), while the final reductions are formalised in the process-local level of BIF calls, with concrete BIF names, as shown in~\Cref{fig:localsem2} below.

In the following, we make a brief description of the process-local reduction rules. The process identifiers in the reduction rules are propagated from the inter-process semantics via actions. First we detail the rule for sequential steps, and the rules for signal arrival (Figure~\ref{fig:localsem1}):

\begin{itemize}
\item \ref{OS:seq} lifts the computational layer to the process-local level. This is the sequential ($\internal$) reduction rule of the semantics. In this rule, the computational layer could be replaced by any other frame stack semantics, such as a semantics for Erlang.
\item \ref{OS:msg} describes message arrival. Whenever a message arrives, it is appended to the mailbox of the process.
\item \ref{OS:exitdrop} describes when should an exit signal be dropped without modifying the state of the process~\cite[Receiving Exit Signals]{erlangrefprocs}.
\item \ref{OS:exitterm} describes when an exit signal terminates the process~\cite[Receiving Exit Signals]{erlangrefprocs}. The process becomes a terminated process by pairing the exit reason with the linked process identifiers. When an exit signal was sent explicitly, and the reason was \texttt{'kill'}\footnote{The \texttt{'kill'} reason causes unconditional termination almost always. We explain the only exception in \Cref{sec:test} with \Cref{ex:kill}.}, it also has to be converted to \texttt{'killed'} for the links (to prevent unnecessary termination of additional processes that are trapping exits).
\item \ref{OS:exitconv} describes when an exit signal should be converted to a message and appended at the end of the mailbox~\cite[Receiving Exit Signals]{erlangrefprocs}. This action can only occur when the \verb|'trap_exit'| flag of the process is set.
\item \ref{OS:linkarr}, \ref{OS:unlinkarr} rules describe arrival of link and unlink signals. In the first case, a process identifier is added to the links of the process, while in the second case, all occurrences of the process identifier are removed from the links.
\end{itemize}

\begin{figure}[t!]
\centering


\begin{equation}
\begin{prooftree}
\hypo{\pstep{(\texttt{call '!'}(\iota_2,\Box)::K,v,q,\textit{pl},b)}{\send{\iota_1}{\iota_2}{\msg{v}}}{(K,v,q,\textit{pl},b)}}
\end{prooftree}
\label{OS:send}
\tag{\textsc{Send}}
\end{equation}

\begin{equation}
\begin{prooftree}
\hypo{\pstep{(\texttt{call 'exit'}(\iota_2,\Box)::K,v,q,\textit{pl},b)}{\send{\iota_1}{\iota_2}{\exit{v}{\textit{ff}}}}{(K,\texttt{'true'},q,\textit{pl},b)}}
\end{prooftree}
\label{OS:exit}
\tag{\textsc{Exit}}
\end{equation}

\begin{equation}
\begin{prooftree}
\hypo{\pstep{(\texttt{call 'link'}(\Box)::K,\iota_2,q,\textit{pl},b)}{\send{\iota_1}{\iota_2}{\link}}{(K,\texttt{'ok'},q,\iota_2::\textit{pl},b)}}
\end{prooftree}
\label{OS:link}
\tag{\textsc{Link}}
\end{equation}

\begin{equation}
\begin{prooftree}
\hypo{\pstep{(\texttt{call 'unlink'}(\Box)::K,\iota_2,q,\textit{pl},b)}{\send{\iota_1}{\iota_2}{\unlink}}{(K,\texttt{'ok'},q,\rem{\iota_2}{\textit{pl}},b)}}
\end{prooftree}
\label{OS:unlink}
\tag{\textsc{Unlink}}
\end{equation}

\begin{equation}
\begin{prooftree}
\hypo{\pstep{(\iota_2,v)::\textit{pl}}{\send{\iota_1}{\iota_2}{\exit{v}{\textit{tt}}}}{\textit{pl}}}
\end{prooftree}
\label{OS:dead}
\tag{\textsc{Dead}}
\end{equation}

\caption{Process-local semantics (part 2)}
\label{fig:localsem2}

\end{figure}

\noindent
Next, we describe the formal rules of signal sending (\Cref{fig:localsem2}):

\begin{itemize}
\item \ref{OS:send} describes message sending. If the BIF \texttt{'!'} is on the top of the frame stack with the target process identifier, and the message is evaluated to a value, a send action is emitted containing the source (which is propagated from the inter-process semantics in the \ref{OS:nsend} rule) and target identifiers and the message value, while the send expression itself is reduced to the message value.
\item \ref{OS:exit} describes explicitly sending an exit signal to a process. If the two-parameter \texttt{'exit'} BIF is on the top of the frame stack with the target process identifier, and the reason is evaluated to a value, an exit action is emitted with the source (which is propagated from the inter-process semantics in \ref{OS:nsend}), target identifiers, and the exit reason value, while the expression is reduced to \texttt{'true'}. Note that when sending an explicit exit signal, the link flag of the signal is false.
\item \ref{OS:link}, \ref{OS:unlink} rules both reduce the evaluable expression to \texttt{'ok'}. In the first case, a link signal is emitted with the source and target identifier, and the target is appended to the links of the process. In the second case, an unlink signal is emitted with the source and target identifier, and the target is removed from the links of the process.
\item \ref{OS:dead} describes the communication of a terminated process. In this rule, the first item of the links of the dead process is removed while an exit signal is emitted to the target with the reason that is specified in this first item. Note that the link flag of this exit signal is \emph{true}, because this exit is sent through a link.
\end{itemize}

\noindent
Finally, we detail the rest of the process-local rules (Figure~\ref{fig:localsem3}):

\begin{itemize}
\item \ref{OS:self} receives the identifier of the process from the inter-process semantics, and evaluates the \texttt{'self'} BIF call to this identifier.
\item \ref{OS:spawn} describes process creation. The spawned process receives its identifier from the inter-process semantics, and this identifier will be the result of this rule. Note that it is necessary that the first parameter of the \texttt{'spawn'} is a function value, while the second is a correct, object-level parameter list (which is checked in the inter-process semantics).
\item \ref{OS:receive} describes message processing. With pattern matching, the first (oldest) message is selected from the mailbox of the process that matches any clause of the \texttt{receive} expression (if more patterns are matching to the same message, the first matching clause is selected). The evaluation continues with the body expression of the selected clause, substituted by the result (pattern variable - value) bindings.
\item \ref{OS:flag} describes when the process flag \verb|'trap_exit'| changes. The result of this rule is the original value of the flag.
\item \ref{OS:term} describes normal termination, i.e. there are no more continuations in the frame stack, and the evaluable expression has already been reduced to a value. The result is a dead process, which will send exit signals to its links with the reason \texttt{'normal'}.
\item \ref{OS:exitself} describes the call of the single-parameter \texttt{'exit'} BIF. It immediately terminates the process, and exit signals will be sent to the linked processes with the parameter reason value. We note that when introducing exceptions in the future, this version of exit signals will be capable of being caught by exception handlers.
\end{itemize}

\begin{figure}

\begin{equation}
\begin{prooftree}
\hypo{\pstep{(\texttt{call\ } \Box()::K,\texttt{'self'},q,\textit{pl},b)}{\self{\iota}}{(K,\iota,q,\textit{pl},b)}}
\end{prooftree}
\label{OS:self}
\tag{\textsc{Self}}
\end{equation}

\begin{equation}
\begin{prooftree}
\hypo{f = \fun{f/k}{x_1}{x_k}{e}}
\infer1{\pstep{(\texttt{call 'spawn'}(f,\Box)::K,\textit{vs},q,\textit{pl},b)}{\spawn{\iota}{f}{\textit{vs}}}{(K,\iota,q,\textit{pl},b)}}
\end{prooftree}
\label{OS:spawn}
\tag{\textsc{Spawn}}
\end{equation}

\begin{equation}
\begin{prooftree}
\hypo{l = \textit{match}(p_i, v)}
\infer[no rule]1{\textit{is\_match}(p_i, v)}
\infer[no rule]1{q = [v_1, \dots, v_n, v, \dots]}
\hypo{\forall j < i: \neg \textit{is\_match}(p_j, v)}
\infer[no rule]1{(\forall m, j: 1 \le m \le k \land 1 \le j \le n \implies \neg \textit{is\_match}(p_m, v_j))}
\infer2{\pstep{(K,\texttt{receive}\ p_1 \rightarrow e_1; \dots; p_k \rightarrow e_k \ \texttt{end},q,\textit{pl},b)}{\receive{v}}{(K,e_i[l],\pop{v}{q},\textit{pl},b)}}
\end{prooftree}
\label{OS:receive}
\tag{\textsc{Receive}}
\end{equation}

\begin{equation}
\begin{prooftree}
\hypo{\cvt{v} = \textit{Some } v'}
\hypo{v'' = \cvt{b}}
\infer2{\pstep{(\texttt{call 'process\_flag'}(\texttt{'trap\_exit'},\Box)::K,v,q,\textit{pl},b)}{\setflag}{(K,v'',q,\textit{pl},v')}}
\end{prooftree}
\label{OS:flag}
\tag{\textsc{Flag}}
\end{equation}


\begin{equation}
\begin{prooftree}
\hypo{\pstep{(\idstack,v,q,\textit{pl},b)}{\term}{\textit{map}\ (\lambda \iota \Rightarrow (\iota,\texttt{'normal'}))\ \textit{pl}}}
\end{prooftree}
\label{OS:term}
\tag{\textsc{Term}}
\end{equation}

\begin{equation}
\begin{prooftree}
\hypo{\pstep{(\texttt{call 'exit'}()::K,v,q,\textit{pl},b)}{\term}{\textit{map}\ (\lambda \iota \Rightarrow (\iota,v))\ \textit{pl}}}
\end{prooftree}
\label{OS:exitself}
\tag{\textsc{ExitSelf}}
\end{equation}

\caption{Process-local semantics (part 3)}
\label{fig:localsem3}

\end{figure}

\subsection{Inter-Process Semantics}\label{sec:inter}

In this section we discuss the inter-process reduction rules for the semantics. The advantage of this formalisation is that the dynamic semantics of the system is described by only 5 rules (by combining the rules from related work~\cite{harrison2017coerl,lanese2019playing,fredlund2001framework} with the same premises but different actions), which resulted in shorter proofs. 
First, we introduce the necessary concepts.

\begin{definition}[Ether]
An ether (denoted by $\Delta$) is a mapping of source and target identifier pairs to lists of signals. We use $\emptyset$ to denote the empty ether, which maps everything to the empty list\footnote{In the implementation, we formalised the ether as a function which maps (source) process identifiers to a function mapping (target) process identifiers to a list of signals.}.
\end{definition}

We use an ether to express non-atomic signal passing (unlike a number of related works~\cite{harrison2017coerl, fredlund2001framework}); that is, the signals sent from one process do not arrive immediately, but they are transferred via the ether. This is also described in the reference manual~\cite{erlangrefprocs}: ``The amount of time that passes between the time a signal is sent and the arrival of the signal at the destination is unspecified but positive''.

In addition, the ether is used to implement the signal ordering guarantee~\cite{erlangrefprocs}, that is ``if an entity sends multiple signals to the same destination entity, the order is preserved''. If a source sends multiple signals to the same target, these signals will be appended to the end of the list associated with the source and target in the ether. However, if multiple processes send signals to the same destination, the arrival order of these signals is not specified, thus they are included in separate lists in the ether based on their source.

\begin{definition}[Node]
A node is a pair ($(\Delta, \Pi) \in \textit{Node}$) of an ether and a process pool. The process pool (denoted by $\Pi$) is a mapping that associates process identifiers with processes. We denote nodes with $\Sigma$ and the empty process pool with $\emptyset$.
\end{definition}

\noindent
On top of these concepts, we introduce notations and metatheoretical functions:

\begin{itemize}
\item $\iota : p \parallel \Pi$: Appends process $p$ associated with the identifier $\iota$ to the process pool $\Pi$. In formalising this we used function update, so that the order of identifiers is irrelevant. Because of this, we are justified in abusing the notation somewhat when we write the rules using pattern matching: without loss of generality, we assume that the item of interest appears in the head position of the collection of processes given.
\item $\remfst{\Delta}{\iota}{\iota'}$: Removes the first element in the ether $\Delta$ from the list associated with $\iota$ source and $\iota'$ destination, and returns a pair of this removed signal and the result ether inside $\textit{Some}$. If the associated list was empty, it returns $\textit{None}$.
\item $\ins{\Delta}{\iota}{\iota'}{s}$: Creates an ether by appending the signal $s$ to the end of the list associated with $\iota$ source and $\iota'$ destination in the ether $\Delta$ (while keeping other parts of $\Delta$ unchanged).
\item $\Pi \setminus \iota$: Creates a process pool by removing the process associated with $\iota$ from process pool $\Pi$. This operation was also formalised by function updates.
\item $\iota \notin \Pi$: Checks whether there is no process associated with $\iota$ in $\Pi$.
\item $\cvtl{vs}$: Creates a metatheoretical list of expressions based on an object-level Core Erlang list (constructed with \texttt{[\_|\_]} and \texttt{[]}); if successful the result is wrapped with a \emph{Some} constructor; if not, \emph{None} is returned.
\end{itemize}

Next, we define the semantics in~\Cref{fig:interprocsem}. This one-step reduction is denoted by $\nstep{\Sigma}{\iota}{a}{\Sigma'}$ that means the node $\Sigma$ is reduced to $\Sigma'$ by taking a reduction step determined by the action $a$ with the process identified by $\iota$. The rules always include a ``first'' process ($\iota : p \parallel \Pi$), nevertheless, any process from the pool can take this place, since any process in a $\parallel$ chain can be the outermost one, as mentioned before. We give a brief, informal description of the inter-process rules now:

\begin{figure}[t!]

\centering

\begin{equation}
\begin{prooftree}
\hypo{\pstep{p}{\send{\iota_1}{\iota_2}{s}}{p'}}
\infer1{\nstep{(\Delta, \iota_1 : p \parallel \Pi)}{\iota_1}{\send{\iota_1}{\iota_2}{s}}{(\ins{\Delta}{\iota_1}{\iota_2}{s}, \iota_1 : p' \parallel \Pi)}}
\end{prooftree}
\label{OS:nsend}
\tag{\textsc{NSend}}
\end{equation}

\begin{equation}
\begin{prooftree}
\hypo{\pstep{p}{\arrive{\iota_1}{\iota_2}{s}}{p'}}
\hypo{\remfst{\Delta}{\iota_1}{\iota_2}= \textit{Some}\ (s, \Delta')}
\infer2{\nstep{(\Delta, \iota_1 : p \parallel \Pi)}{\iota_1}{\arrive{\iota_1}{\iota_2}{s}}{(\Delta', \iota_1 : p' \parallel \Pi)}}
\end{prooftree}
\label{OS:narrive}
\tag{\textsc{NArrive}}
\end{equation}

\begin{equation}
\begin{prooftree}
\hypo{\nstep{(\Delta, \iota : [] \parallel \Pi)}{\iota}{\term}{(\Delta, \Pi \setminus \iota)}}
\end{prooftree}
\label{OS:nterm}
\tag{\textsc{NTerm}}
\end{equation}

\begin{equation}
\begin{prooftree}
\hypo{\iota_2 \notin (\iota_1 : p \parallel \Pi)}
\infer[no rule]1{\pstep{p}{\spawn{\iota_2}{v}{\textit{vs}}}{p'}}
\hypo{v = \fun{f/k}{x_1}{x_k}{e}}
\infer[no rule]1{\cvtl{vs} = \textit{Some } [v_1, \dots, v_k]}
\infer2{\nstep{(\Delta, \iota_1 : p \parallel \Pi)}{\iota_1}{\spawn{\iota_2}{v}{\textit{vs}}}{(\Delta, \iota_2 : ([], \mathtt{apply}\ v(v_1, \dots, v_k), [], [], \textit{ff}) \parallel \iota_1 : p' \parallel \Pi)}}
\end{prooftree}
\label{OS:nspawn}
\tag{\textsc{NSpawn}}
\end{equation}

\begin{equation}
\begin{prooftree}
\hypo{\pstep{p}{a}{p'}}
\hypo{a \in \{\self{\iota}, \term, \internal, \setflag\} \cup \{\receive{v} \mid v \in Value \}}
\infer2{\nstep{(\Delta, \iota : p \parallel \Pi)}{\iota}{a}{(\Delta, \iota : p' \parallel \Pi)}}
\end{prooftree}
\label{OS:nother}
\tag{\textsc{NOther}}
\end{equation}

\caption{Formal semantics of communication between processes}
\label{fig:interprocsem}

\end{figure}

\begin{itemize}
\item \ref{OS:nsend} describes signal sending. While a process with the identifier $\iota$ is reduced by emitting a send action, the contents of this action (target, source, and signal) are placed into the ether.
\item \ref{OS:narrive} describes how an element is (nondeterministically) removed from the ether. Any signal can be removed from the lists in the ether, if the signal is the first element of that list, and there is a live process with the destination identifier in the process pool.
\item \ref{OS:nterm} describes how a process identifier is freed. When a dead process has notified all of its links, its identifier is removed from the process pool.
\item \ref{OS:nspawn} describes the creation of a new process. The new process is assigned a non-used identifier, and it starts evaluating the function application described in the spawn action of the rule (note that conversion from object-level to meta-level lists is needed). The initial configuration of the new process is the empty frame stack (continuation), the given function application as the evaluable expression, empty mailbox, it has no links, and it does not trap exit signals.
\item \ref{OS:nother} describes the reduction in case of any other action, that is, this rule propagates these actions to the process-local level.
\end{itemize}

We note that in every rule of this semantics, exactly one process is reduced. Furthermore, all reduction rules (except \ref{OS:nterm}) actually propagate the action to the process-local semantics, while modifying the ether or the process pool. We also introduce the following notations on top of the inter-process semantics:

\begin{itemize}
\item $\nstepk{\Sigma}{l}{\Sigma'}$ denotes a special reflexive, transitive closure of the relation $\xrightarrow[]{\iota : a}$, which traces the actions in the list $l$ in forms of $(\iota, a)$ pairs. We use $\nstepk{\Sigma}{\dots}{\Sigma'}$ when the trace is not relevant. For example, if a node $\Sigma$ can be reduced to $\Sigma'$ in the three following steps: 1) the process identified by $\iota$  sends a message $v$ to the process identified by $\iota'$, 2) this message arrives to the target, 3) the message is received by the target, we use
\begin{align*}
\nstepk{\Sigma}{[(\iota,\send{\iota}{\iota'}{\msg{v}}), (\iota',\arrive{\iota}{\iota'}{\msg{v}}), (\iota',\receive{v})]}{\Sigma'}.
\end{align*}
\item $\nstepstar{\Sigma}{\Sigma}$ denotes a reduction sequence from node $\Sigma$ to node $\Sigma'$ that contains only sequential ($\internal$) reduction steps.
\end{itemize}

\paragraph{Discussion.} There are other approaches (e.g.~\cite{fredlund2001framework}) which define fewer actions for the semantics by defining input, output, spawn, and silent actions. In these approaches, $\receive{v}, \setflag, \term, \internal$ could all be handled as silent actions, since they affect the state of a single process and do not communicate. Still, we formalised the more fine-grained version, because this allows us to group the rules of the semantics into more categories. By coupling the aforementioned actions, the less detailed approach can also be simulated. Moreover, we proved theorems (specifically, \Cref{thm:equiv}) which would not be provable if other actions were also considered to be silent. 

\section{Semantics Validation}\label{sec:validation}

After defining a formal semantics, the next step is to validate it~\cite{blazy2009mechanized}. We use two approaches: 1) we evaluate simple parallel programs and compare the results to the results of the Erlang/OTP compiler, and 2) we prove properties of the semantics. We are also investigating ways in which the concurrent semantics can be executed efficiently, which is a necessary step to enable extensive validation against the reference implementation.

\begin{figure}[h!]
    \centering
\tikzset{main node/.style={circle,draw,minimum size=1cm}}
\begin{tikzpicture}
  \node[main node] (1) {$1$};
  \node[main node] (2) [below left = 2cm and 1.5cm of 1]  {$2$};
  \node[main node] (3) [below right = 2cm and 1.5cm of 1] {$3$};
  \path[->,draw,thick]
  (1) edge node [text width=0.9cm,left] {\texttt{'fst'}} (2)
  (2) edge node [below] {\texttt{'fst'}} (3)
  (1) edge node [right] {\ \texttt{'snd'}} (3);
\end{tikzpicture}
    \caption{Actor diagram for \Cref{ex:signals}}
    \label{fig:exampleprocs}
\end{figure}

\subsection{Example Program Evaluation}\label{sec:test}

In this section we present some simple program evaluation case studies that demonstrate how the semantics operates.

\begin{example}[Signal ordering]\label{ex:signals}

The first example  illustrates when the signal ordering guarantee cannot be applied. Let us consider three processes (with the identifiers $1$, $2$, $3$), which evaluate the following expressions.

\begin{enumerate}
\item \texttt{let X = call '!'(2, 'fst') in call '!'(3, 'snd')}
\item \texttt{receive X -> call '!'(3, X) end}
\item \texttt{receive X -> X end}
\end{enumerate}

Next, we construct a node with the empty ether from these processes, and start evaluating it. We use $\Pi$, to denote the process pool constructed from 2 and 3. For simplicity, we omit the list of linked processes and the trap flag, since they are not used during this evaluation. First, we reduce process $1$, since the others are all blocked by \texttt{receive} expressions. This evaluation puts the two messages into the ether.

\begin{flalign}
\begin{split}
&\nstepkbr{(\emptyset, 1: (\idfs, \texttt{let X = call '!'(2, 'fst') in call '!'(3, 'snd')}, []) \parallel \Pi)}{\dots}
          {(\ins{\ins{\emptyset}{1}{2}{\msg{\texttt{'fst'}}}}{1}{3}{\msg{\texttt{'snd'}}}, 1: (\idfs, \texttt{'snd'}, []) \parallel \Pi)}
\end{split}\label{example:step1}\tag{\emph{Init}}
\end{flalign}

We denote the result process pool with $\Pi_1$ without process $3$. Next, we can evaluate process $3$, to which the message \texttt{'snd'} arrives. Then the \texttt{receive} expression removes it from the mailbox and processes it. Thus the final value upon termination in process 3 is \texttt{'snd'}.

\begin{flalign*}
&\nstepbr{(\ins{\ins{\emptyset}{1}{2}{\msg{\texttt{'fst'}}}}{1}{3}{\msg{\texttt{'snd'}}}, \\&\quad\ 3: (\idfs, \texttt{receive X -> X end}, []) \parallel \Pi_1)}{3}{\arrive{1}{3}{\msg{\texttt{'snd'}}}}
{(\ins{\emptyset}{1}{2}{\msg{\texttt{'fst'}}}, 3: (\idfs, \texttt{receive X -> X end}, [\texttt{'snd'}]) \parallel \Pi_1)}\xrightarrow[]{\dots}\!^*\\
&(\ins{\emptyset}{1}{2}{\msg{\texttt{'fst'}}}, 3: (\idfs, \texttt{'snd'}, []) \parallel \Pi_1)
\end{flalign*}

Note that the last configuration we presented above could still progress, because process $2$ can receive and forward the message \texttt{'fst'}.

However, this was not the only option to evaluate this simple program. Instead of evaluating process $3$ in the previous reductions, we can progress with process $2$ (from the state reached in \ref{example:step1}). We denote the process pool containing the terminated process 1 and process 3  with $\Pi_2$. First, the message \texttt{'fst'} arrives to process $2$ which removes it from the mailbox, and forwards it to process $3$.

\begin{flalign*}
&\nstepbr{(\ins{\ins{\emptyset}{1}{2}{\msg{\texttt{'fst'}}}}{1}{3}{\msg{\texttt{'snd'}}}, \\&\quad\ 2: (\idfs, \texttt{receive X -> call '!'(3, X) end}, []) \parallel \Pi_2)}{2}{\arrive{1}{2}{\msg{\texttt{'fst'}}}}
{(\ins{\emptyset}{1}{3}{\msg{\texttt{'snd'}}}, \\&\quad\ 2: (\idfs, \texttt{receive X -> call '!'(3, X) end}, [\texttt{'fst'}]) \parallel \Pi_2)}\xrightarrow[]{\dots}\!^*\\
&(\ins{\ins{\emptyset}{1}{3}{\msg{\texttt{'snd'}}}}{2}{3}{\msg{\texttt{'fst'}}}, 2: (\idfs, \texttt{'fst'}, []) \parallel \Pi_2)
\end{flalign*}

After processes 1 and 2 are terminated (we denote the pool containing these with $\Pi_3$), we evaluate process 3. At this point, either of the messages in the ether could arrive first at process $3$, which will be processed then by the \texttt{receive} expression, since their source is different. We present the case when \texttt{'fst'} arrives first.

\begin{flalign*}
&\nstepbr{(\ins{\ins{\emptyset}{1}{3}{\msg{\texttt{'snd'}}}}{2}{3}{\msg{\texttt{'fst'}}}, \\&\quad\ 3: (\idfs, \texttt{receive X -> X end}, []) \parallel \Pi_3)}{3}{\arrive{2}{3}{\msg{\texttt{'fst'}}}}
{(\ins{\emptyset}{1}{3}{\msg{\texttt{'snd'}}}, \\&\quad\ 3: (\idfs, \texttt{receive X -> X end}, [\texttt{'fst'}]) \parallel \Pi_3)}
 \xrightarrow[]{3 : \arrive{1}{3}{\msg{\texttt{'snd'}}}} \\
&(\emptyset, 3: (\idfs, \texttt{receive X -> X end}, [\texttt{'fst'}, \texttt{'snd'}]) \parallel \Pi_3)
\xrightarrow[]{\dots}\!^*\\
&(\ins{\emptyset}{1}{2}{\msg{\texttt{'snd'}}}, 3: (\idfs, \texttt{'fst'}, [\texttt{'snd'}]) \parallel \Pi_3)
\end{flalign*}

However, with this reduction sequence, process 3 terminates with \texttt{'fst'}. The signal ordering guarantee was not applicable in this scenario, because the messages that process 3 received are from different sources.
\hfill$\triangle$
\end{example}

\begin{example}[Exit signals]\label{ex:kill}
Next, we present an example about sending exit signals, specifically we show the difference between the one- and two-parameter \texttt{'exit'} BIFs. Consider two processes:

\begin{enumerate}
\item \texttt{let X = call 'link'(2) in call 'exit'(1, 'kill')}
\item \texttt{receive X -> X end}
\end{enumerate}

The second process is set to trap exit signals. Once again, we start the evaluation with the first process. In the first steps, process 1 creates the link between the two processes:

\begin{flalign}
\begin{split}
&\nstepkbr{(\emptyset, 1: (\idfs, \texttt{let X = call 'link'(2) in call 'exit'(1, 'kill')}, [], [], \textit{ff}) \parallel 
                       \\&\quad\  2: (\idfs, \texttt{receive X -> X end}, [], [], \textit{tt}) \parallel \emptyset)}{\dots}
            {(\ins{\emptyset}{1}{2}{\link}, 1: (\idfs, \texttt{call 'exit'(1, 'kill')}, [], [2], \textit{ff}) \parallel 
                       \\&\qquad\qquad\qquad\quad\,  2: (\idfs, \texttt{receive X -> X end}, [], [], \textit{tt}) \parallel \emptyset)} \xrightarrow[]{2 : \arrive{1}{2}{\link}} \\
&(\emptyset, 1: (\idfs, \texttt{call 'exit'(1, 'kill')}, [], [2], \textit{ff}) \parallel 
                        \\&\quad\  2: (\idfs, \texttt{receive X -> X end}, [], [1], \textit{tt}) \parallel \emptyset)
\end{split}\label{step:link}\tag{\emph{Link}}
\end{flalign}

Next, the first process terminates itself with the two-parameter \texttt{'exit'}. This involves multiple reduction steps, because the signal needs to be put into the ether, and then retrieved from it. Then the reason will be converted to \texttt{'killed'} because the two-parameter \texttt{'exit'} always sets the link flag of the exit signal to \textit{ff}.

\begin{flalign*}
&\nstepkbr{(\emptyset, 1: (\idfs, \texttt{call 'exit'(1, 'kill')}, [], [2], \textit{ff}) \parallel 
                        \\&\quad\  2: (\idfs, \texttt{receive X -> X end}, [], [1], \textit{tt}) \parallel \emptyset)}{\dots}
{(\ins{\emptyset}{1}{1}{\exit{\texttt{'kill'}}{\textit{ff}}}, 1: (\idfs, \texttt{'true'}, [], [2], \textit{ff}) \parallel 
                       \\&\quad\  2: (\idfs, \texttt{receive X -> X end}, [], [1], \textit{tt}) \parallel \emptyset)} \xrightarrow[]{1 : \arrive{1}{1}{\exit{\texttt{'kill'}}{\textit{ff}}}} \\
&(\emptyset, 1: [(2, \texttt{'killed'})] \parallel 2: (\idfs, \texttt{receive X -> X end}, [], [1], \textit{tt}) \parallel \emptyset)
\end{flalign*}

Next, we propagate the exit signal through the link, and it will be converted to a message because of the trap flag in the execution of process 2.

\begin{flalign*}
&(\emptyset, 1: [(2, \texttt{'killed'})] \parallel 
                         2: (\idfs, \texttt{receive X -> X end}, [], [1], \textit{tt}) \parallel \emptyset) \xrightarrow[]{\dots}\!^* \\
&            (\ins{\emptyset}{1}{2}{\exit{\texttt{'killed'}}{\textit{tt}}}, 
                      \\&\quad\    2: (\idfs, \texttt{receive X -> X end}, [], [1], \textit{tt}) \parallel \emptyset) \xrightarrow[]{2 : \arrive{1}{2}{\exit{\texttt{'killed'}}{\textit{tt}}}} \\
&(\emptyset, 2: (\idfs, \texttt{receive X -> X end}, [\texttt{['EXIT', 1, 'killed']}], [1], \textit{tt}) \parallel \emptyset)\xrightarrow[]{\dots}\!^* \\
&(\emptyset, 2: (\idfs, \texttt{['EXIT', 1, 'killed']}, [], [1], \textit{tt}) \parallel \emptyset
\end{flalign*}

However, if we use the single parameter \texttt{'exit'} BIF, the reduction would be carried out otherwise. We start the evaluation from the analogous state to the point~\ref{step:link} above. It immediately terminates the process without sending signals into the ether. This also causes the reason \texttt{'kill'} not to be converted to \texttt{'killed'}. Next, this exit signal will be sent through a link (the link flag of the signal is \textit{tt}), which enables the use of \ref{OS:exitconv} in process 2.

\begin{flalign*}
&\nstepkbr{(\emptyset, 1: (\idfs, \texttt{call 'exit'('kill')}, [], [2], \textit{ff}) \parallel 
                   \\&\quad\  2: (\idfs, \texttt{receive X -> X end}, [], [1], \textit{tt}) \parallel \emptyset)}{\dots}
             {(\emptyset, 1: [(2, \texttt{'kill'})] \parallel 
                         2: (\idfs, \texttt{receive X -> X end}, [], [1], \textit{tt}) \parallel \emptyset)} \xrightarrow[]{\dots}\!^*\\
&            (\ins{\emptyset}{1}{2}{\exit{\texttt{'kill'}}{\textit{tt}}}, 
                        \\&\quad\ 2: (\idfs, \texttt{receive X -> X end}, [], [1], \textit{tt}) \parallel \emptyset) \xrightarrow[]{2 : \arrive{1}{2}{\exit{\texttt{'kill'}}{\textit{tt}}}} \\
&(\emptyset, 2: (\idfs, \texttt{receive X -> X end}, [\texttt{['EXIT', 1, 'kill']}], [1], \textit{tt}) \parallel \emptyset)\xrightarrow[]{\dots}\!^* \\
&(\emptyset, 2: (\idfs, \texttt{['EXIT', 1, 'kill']}, [], [1], \textit{tt}) \parallel \emptyset
\end{flalign*}

We should note that the \texttt{'kill'} reason is normally used to terminate a process regardless of its current state. Although, in this case (when the signal is sent through a link, i.e. its link flag is set) \texttt{'kill'} does \emph{not} terminate the process in question.
\hfill$\triangle$
\end{example}

\subsection{Properties of the Semantics}\label{sec:property}

After formally evaluating simple programs, we proved some fundamental properties of the layers of the semantics, and formalised the proofs in the Coq theorem prover~\cite{coreerlangmini}. In this section we highlight the most important properties, and provide sketches of the proofs; for more insights, we refer to the formalisation. First, we show the determinism of the sequential and process-local levels.

\begin{theorem}[Sequential and process-local evaluation is deterministic]\label{thm:determinism}
For all frame stacks $K, K', K''$ and expressions $e, e', e''$, if $\rewrites{K}{e}{K'}{e'}$ and $\rewrites{K}{e}{K''}{e''}$, then $K' = K''$ and $e' = e''$.

Similarly, for all processes $p, p', p''$, and actions $a$, if $\pstep{p}{a}{p'}$ and $\pstep{p}{a}{p''}$, then $p' = p''$.
\end{theorem}
\begin{proof}
To prove determinism (in both semantics), we carried out case distinction based on the two reduction premises. If both use the same reduction rule, their result is equal, otherwise a contradiction is found between the premises of the different rules.
\end{proof}

The determinism of these layers of the semantics is a natural property; one process should handle an incoming action in the same way in the same inner state. However, we found that the conditions in the reference manual~\cite[Receiving Exit Signals]{erlangrefprocs} are ambiguous in the way that they describe how to handle exit signals. We checked with the reference implementation what the correct conditions are, and encoded them in the premises for reduction rules about exit signals: \ref{OS:exitconv}, \ref{OS:exitdrop}, and \ref{OS:exitterm}.

In the previous sections, we emphasised why the concept of the ether is necessary to ensure the signal ordering guarantee. We formally verified this property.

\begin{theorem}[Signal ordering guarantee]
For all nodes $\Sigma_1, \Sigma_2, \Sigma_3$, process identifiers $\iota, \iota'$, and unique signals\footnote{They are different from any other signal in the starting configuration.} $s_1 \neq s_2$, if $\nstep{\Sigma_1}{\iota}{\send{\iota}{\iota'}{s_1}}{\Sigma_2}$ and $\nstep{\Sigma_2}{\iota}{\send{\iota}{\iota'}{s_2}}{\Sigma_3}$, then for all nodes $\Sigma_4$ and action traces $l$ which satisfy $\nstepk{\Sigma_3}{l}{\Sigma_4}$ and also $(\iota', \arrive{\iota}{\iota'}{s_1}) \notin l$ there is no node $\Sigma_5$ at which $s_2$ can arrive: $\nstep{\Sigma_4}{\iota'}{\arrive{\iota}{\iota'}{s_2}}{\Sigma_5}$.

\end{theorem}
\begin{proof}
We proved this theorem by induction on the length of the reduction chain $\nstepk{\Sigma_3}{l}{\Sigma_4}$. In the base case, the first removable element in the ether is either $s_1$ or another signal which is different from $s_2$. In the inductive case, we suppose that there is a reduction chain of length $k$ which does not remove $s_1$ from the ether. Then there is the $(k+1)$th reduction step, which also cannot remove $s_1$ from the ether, based on the hypotheses. Thus once again, the first removable element from the ether is either $s_1$ or another signal which is different from $s_2$.
\end{proof}

This theorem informally states the following: if two signals have been sent from the same sender to the same target, after taking any number of reduction steps, which do not contain the arrival of the first signal, it is not possible that the second signal will arrive to the target.

We also proved a number of confluence properties, which are the basis of proving bisimulation-based program equivalence. Our goal is to prove that sequential evaluation ($\nstepstar{\Sigma}{\Sigma'}$) produces equivalent nodes. The first theorem expresses that a sequential reduction can be carried out after another reduction step if this step does not terminate the process. Otherwise, the sequential reduction cannot be executed after the other action. This property holds for both process-local and inter-process semantics.

\begin{theorem}[Confluence of sequential reductions in the same process]\label{thm:internalswap}
For all processes $p_1, p_2, p_2'$, and action $a$, supposing that $\pstep{p_1}{\internal}{p_2}$ and $\pstep{p_1}{a}{p_2'}$, then there exists a process $p_3$ that satisfies $\pstep{p_2}{a}{p_3}$ and $(\pstep{p_2'}{\internal}{p_3} \lor p_2' = p_3)$.

Similarly, for all nodes $\Sigma_1, \Sigma_2, \Sigma_2'$, process identifiers $\iota$, and actions $a$, supposing that $\nstep{\Sigma_1}{\iota}{\internal}{\Sigma_2}$ and $\nstep{\Sigma_1}{\iota}{a}{\Sigma_2'}$, then there exists a node $\Sigma_3$ that satisfies $\nstep{\Sigma_2}{\iota}{a}{\Sigma_3}$ and $(\nstep{\Sigma_2'}{\iota}{\internal}{\Sigma_3} \lor \Sigma_2' = \Sigma_3)$.
\end{theorem}
\begin{proof}
The proof for both semantics relies on case distinction in the derivation of $\nstep{\Sigma_1}{\iota}{a}{\Sigma_2'}$. There are actually two separate cases:
\begin{itemize}
\item If the action $a$ does not terminate the process (still, it potentially modifies either the mailbox, or the list of linked processes, or the \texttt{'trap\_exit'} flag), then the sequential reduction step can be taken after this action too, since these steps are not influenced by the mentioned attributes of the process.
\item If the action $a$ terminates the process, then $p_2'$ and $p_3$ denote the same terminated process, since sequential steps do not modify the list of linked processes, which is the only attribute of a live process that is kept when it terminates.
\end{itemize}
\end{proof}

This theorem is used when two reductions for the same process need to be chained after each other. Actually, this theorem is a stepping stone towards proving \Cref{thm:chainend}.

The next theorem concerns different processes: possible reduction steps can be carried out after each other if they are not both \textit{spawn} actions.

\begin{theorem}[Action ordering]\label{thm:stepchain}
For all nodes $\Sigma_1, \Sigma_2, \Sigma_2'$, process identifiers $\iota \neq \iota'$, actions $a, a'$, which are not both \textit{spawn} actions, if $\nstep{\Sigma_1}{\iota}{a}{\Sigma_2}$ and $\nstep{\Sigma_1}{\iota'}{a'}{\Sigma_2'}$ then there exists a node $\Sigma_3$, which can be reached from $\Sigma_2$ with action $a'$: $\nstep{\Sigma_2}{\iota'}{a'}{\Sigma_3}$.
\end{theorem}
\begin{proof}
This theorem is proved by case separation on the two reduction premises. There is no scheduling algorithm formalised in the semantics, thus any process can be reduced if it is not in a stuck configuration (i.e. if it is waiting for a message to evaluate a \texttt{receive} expression). We can define any order for the reductions of different processes (except if both reductions are labelled by \emph{spawn} actions), because both of the reductions in the premise can always be carried out. The only action $a$ in the first reduction that could prevent making the second reduction (with action $a'$) is the arrival of an exit signal that terminates the process identified by $\iota'$, but the premise $\iota \neq \iota'$ rules this case out.
\end{proof}

The premise that restricts \emph{spawn} actions is necessary because we cannot assure that these spawned processes obtain the same process identifiers if their spawn order is reversed (currently, the semantics assigns fresh process identifiers to spawned processes based on the list of process identifiers already in use). This theorem is also a stepping stone towards \Cref{thm:chainfront}.

The following theorems contain any-step reduction chains. The first of these theorem expresses that if there are $\internal$ actions and an additional action that can be executed in a configuration, then either this additional action can be executed at the final node after executing the chain, or it was $\internal$-reduction inside the chain.

\begin{theorem}[Chaining a reduction to the end of an sequential sequence]\label{thm:chainend}
For all nodes $\Sigma_1, \Sigma_4, \Sigma_4'$, process identifier $\iota$, action $a$, and action traces $l$, which only include internal actions paired with any process identifiers, if $\nstepk{\Sigma_1}{l}{\Sigma_4}$ and $\nstep{\Sigma_1}{\iota}{a}{\Sigma_4'}$, then there are two potential scenarios:

\begin{itemize}
\item Either there is a node $\Sigma_5$ which can be reached by a reduction from $\Sigma_4$: $\nstep{\Sigma_4}{\iota}{a}{\Sigma_5}$.
\item Or $a = \internal$ and there are nodes $\Sigma_2, \Sigma_3$ and action traces $l_1, l_2$, which can be used to split the sequential reduction steps: $\nstepk{\Sigma_1}{l_1}{\Sigma_2}$, $\nstep{\Sigma_2}{\iota}{a}{\Sigma_3}$ and $\nstepk{\Sigma_3}{l_2}{\Sigma_4}$, moreover $l = l_1 \concat [(\iota, a)] \concat l_2$.
\end{itemize}

\end{theorem}
\begin{proof}
We proved this theorem by induction on the reduction chain $\nstepk{\Sigma_1}{l}{\Sigma_4}$. The base case is solved by the premise $\nstep{\Sigma_1}{\iota}{a}{\Sigma_4'}$ (by choosing $\Sigma_5 = \Sigma_4'$), since $\nstepk{\Sigma_1}{l}{\Sigma_4}$ was a 0-step reduction, thus $\Sigma_1 = \Sigma_4$.
In the inductive case, we did case distinction whether $a = \internal$ and $(\iota, a)$ is included in $l$. If this is not true, we can make the reduction determined by $(\iota, a)$ from the configuration $\Sigma_4$ based on the induction hypothesis and \Cref{thm:internalswap}. Otherwise $(\iota, a) = (\iota, \internal)$ is included in the action trace $l$.
\end{proof}

This theorem expresses one of the fundamental properties needed to prove that $\nstepstar{}{}$ is a weak bisimulation, (\Cref{thm:equiv} below).
The next theorem is the other fundamental property required. If there is an action that is executed at the end of the execution of a sequential reduction chain, and it can be executed in the starting configuration too, then from the result of the second derivation the result of the first one can be reached by only sequential steps.

\begin{theorem}[Confluence of sequential reductions]\label{thm:chainfront}
For all nodes $\Sigma_1, \Sigma_2, \Sigma_2', \Sigma_3$, process identifier $\iota$, and action $a$, if $\nstepstar{\Sigma_1}{\Sigma_2}$, and a reduction can be done in the starting and in the final configuration too: $\nstep{\Sigma_1}{\iota}{a}{\Sigma_2'}$, and $\nstep{\Sigma_2}{\iota}{a}{\Sigma_3}$, then $\nstepstar{\Sigma_2'}{\Sigma_3}$.
\end{theorem}
\begin{proof}
The proof of this property is also carried out by induction on the reduction chain $\nstepstar{\Sigma_1}{\Sigma_2}$. In the base case $\Sigma_1 = \Sigma_2$ and by \Cref{thm:determinism}, $\Sigma_2' = \Sigma_3$, while the proof of the inductive case is based on \Cref{thm:stepchain} and the induction hypothesis.
\end{proof}

What if this potentially non-sequential action was the arrival of an exit signal? That will potentially terminate a process, which could have taken some internal steps. We note that with the $\longrightarrow^*$ reduction in the conclusion we do not say that the steps are preserved, thus the result node after the arrival of the exit signal can take fewer internal steps by leaving the steps for the terminated process out.

\section{Program Equivalence}\label{sec:bisim}

In this section, we investigate program equivalence using bisimulation. Bisimulations are relations between nodes that are preserved by the reduction steps.

\begin{definition}[Bisimulation]\label{def:bisim}
A relation $R$ is a bisimulation if it satisfies the following two properties:

\begin{itemize}
\item For all nodes $\Sigma_1, \Sigma_2, \Sigma_1'$, process identifiers $\iota$, and actions $a$, if $(\Sigma_1, \Sigma_2) \in R$ and $\nstep{\Sigma_1}{\iota}{a}{\Sigma_1'}$, then there is a node $\Sigma_2'$, which is reducible from $\Sigma_2$ with the action $a$: $\nstep{\Sigma_2}{\iota}{a}{\Sigma_2'}$, and $(\Sigma_1',\Sigma_2') \in R$.
\item For all nodes $\Sigma_1, \Sigma_2, \Sigma_2'$, process identifiers $\iota$, and actions $a$, if $(\Sigma_1, \Sigma_2) \in R$ and $\nstep{\Sigma_2}{\iota}{a}{\Sigma_2'}$, then there is a node $\Sigma_1'$, which is reducible from $\Sigma_1$ with the action $a$: $\nstep{\Sigma_1}{\iota}{a}{\Sigma_1'}$, and $(\Sigma_1',\Sigma_2') \in R$.
\end{itemize}

\end{definition}

\noindent
We can show that equality satisfies the above conditions of being a bisimulation.

\begin{theorem}
The equality of nodes is a bisimulation.
\end{theorem}
\begin{proof}
This property is just a simple consequence of the definition of bisimulation.
\end{proof}

We also defined a relaxed variant: weak bisimulations omit $\internal$ actions taken in the semantics, so only communication actions should preserve the relation.

\begin{definition}[Weak bisimulation]\label{def:wbisim}
A relation $R$ is a weak bisimulation if it satisfies the following two properties:

\begin{itemize}
\item For all nodes $\Sigma_1, \Sigma_2, \Sigma_1'$, process identifiers $\iota$, and actions $a \neq \internal$, if $(\Sigma_1, \Sigma_2) \in R$ and $\nstep{\Sigma_1}{\iota}{a}{\Sigma_1'}$, then there are nodes $\Sigma_2^1, \Sigma_2^2, \Sigma_2'$, which are reducible from $\Sigma_2$ in the following way: $\nstepstar{\Sigma_2}{\Sigma_2^1}$, $\nstep{\Sigma_2^1}{\iota}{a}{\Sigma_2^2}$, and $\nstepstar{\Sigma_2^2}{\Sigma_2'}$, and $(\Sigma_1',\Sigma_2') \in R$.
\item For all nodes $\Sigma_1, \Sigma_2, \Sigma_1'$, process identifiers $\iota$, and actions $a \neq \internal$, if $(\Sigma_1, \Sigma_2) \in R$ and $\nstep{\Sigma_2}{\iota}{a}{\Sigma_2'}$, then there are nodes $\Sigma_1^1, \Sigma_1^2, \Sigma_1'$, which are reducible from $\Sigma_1$ in the following way: $\nstepstar{\Sigma_1}{\Sigma_1^1}$, $\nstep{\Sigma_1^1}{\iota}{a}{\Sigma_1^2}$, and $\nstepstar{\Sigma_1^2}{\Sigma_1'}$, and $(\Sigma_1',\Sigma_2') \in R$.
\end{itemize}

\end{definition}

\noindent
Bisimulations satisfy the natural property of being weak bisimulations.

\begin{theorem}
Bisimulations are weak bisimulations.
\end{theorem}
\begin{proof}
This property is also a simple consequence of the definitions, since we can choose 0-step reductions for $\nstepstar{\Sigma_2}{\Sigma_2^1}$ and $\nstepstar{\Sigma_2^2}{\Sigma_2'}$ in \Cref{def:wbisim}, while the middle step $\nstep{\Sigma_2^1}{\iota}{a}{\Sigma_2^2}$ is obtained from \Cref{def:bisim}.
\end{proof}

We consider two programs $(\Sigma, \Sigma')$ equivalent if there is a relation $R$ that is a weak bisimulation and $(\Sigma, \Sigma') \in R$.
Next, we prove that sequential evaluation is a weak bisimulation. For this proof we used the chaining properties (\Cref{thm:chainend} and \Cref{thm:chainfront}) of the semantics.

\begin{theorem}\label{thm:equiv}
$\longrightarrow^*$ (between nodes) is a weak bisimulation.
\end{theorem}
\begin{proof}
To avoid ambiguity, we use $\Lambda$ to denote the available nodes in the proof, while we keep $\Sigma$ for the definitions. To prove that a relation is a weak bisimulation, two properties need to be proved:
\begin{itemize}
\item For the first part of \Cref{def:wbisim} we have $\nstepstar{\Lambda_1}{\Lambda_2}$ and $\nstep{\Lambda_1}{\iota}{a}{\Lambda_1'}$ as assumptions. We can chain the reduction determined by $a$ to the end of the sequential reduction sequence by \Cref{thm:chainend} ($\nstep{\Lambda_2}{\iota}{a}{\Lambda_3}$ for some $\Lambda_3$). Actually, the second possible conclusion (i.e. the $a = \internal$) of this theorem can not occur here, because of the restriction $a \neq \internal$ in \Cref{def:wbisim}. We need to prove that $\nstepstar{\Lambda_2}{\Sigma_2^1}$, $\nstep{\Sigma_2^1}{\iota}{a}{\Sigma_2^2}$, $\nstepstar{\Sigma_2^2}{\Sigma_2'}$, and $\nstepstar{\Lambda_1'}{\Sigma_2'}$ for suitable $\Sigma$ nodes. We can choose $\Sigma_2^1 = \Lambda_2$, $\Sigma_2^2 = \Lambda_3$, and $\Sigma_2' = \Lambda_3$, thus the first and second $\longrightarrow^*$ reductions are 0-step reductions, while the single-step reduction is among the assumptions. The reduction $\nstepstar{\Sigma_1'}{\Sigma_2'}$ remains, which can be proved by \Cref{thm:chainfront}.
\item To satisfy the second part of \Cref{def:wbisim} we have $\nstepstar{\Lambda_1}{\Lambda_2}$ and $\nstep{\Lambda_2}{\iota}{a}{\Lambda_2'}$ as assumptions. We need to prove $\nstepstar{\Lambda_1}{\Sigma_1^1}$, $\nstep{\Sigma_1^1}{\iota}{a}{\Sigma_1^2}$, and $\nstepstar{\Sigma_1^2}{\Sigma_1'}$, and $\nstepstar{\Sigma_1'}{\Lambda_2'}$ for suitable nodes. We choose $\Sigma_1^1 = \Lambda_2$, $\Sigma_1^2 = \Lambda_2'$, and $\Sigma_1' = \Lambda_2'$, thus the first two reductions are among the assumptions, while the third and fourth ones are 0-step reductions.
\end{itemize}
\end{proof}

With this proof, we can state that a node is equivalent to the nodes to which it reduces by using sequential steps only. For example, we can derive the following property:

\newcommand{\mm}{\mytt{letrec\ } \texttt{'mm'/2} = mm \mytt{ in apply\ 'mm'/2}(f, \mytt{[}\mytt{0}\mytt{,}\mytt{1}\mytt{,}\mytt{2}\mytt{]})}
\begin{example}
For all nodes $\Pi$, ethers $\Delta$, frame stacks $K$, process identifiers $\iota$, mailboxes $\textit{q}$, list of process identifiers $\textit{pl}$, and process flags $\textit{flag}$, the following nodes are equivalent (where \textit{mm} denotes the function expression inside \verb|letrec| in \Cref{example:cerlmap}, and $f$ denotes the successor function from \Cref{ex:seqeval}).
\[
(\Delta, \iota : (K, \mm, \textit{q}, \textit{pl}, \textit{flag}) \parallel \Pi)
\]
\[
(\Delta, \iota : (K, \mytt{[}\mytt{1}\mytt{,}\mytt{2}\mytt{,}\mytt{3}\mytt{]}, \textit{q}, \textit{pl}, \textit{flag}) \parallel \Pi)
\]
\end{example}
\begin{proof}
    We have already shown in \Cref{ex:seqeval} how the complex \texttt{letrec} expression can be reduced to a list of values. Using this fact together with \Cref{thm:equiv} we can prove this equivalence (note that the sequential steps of the evaluation can be lifted to the inter-process level with rules \ref{OS:seq} and \ref{OS:nother}).
\end{proof}

There is a natural question, whether any evaluation sequence could be proved to be a bisimulation. Unfortunately, that is not the case.

\begin{theorem}
For all $l$ action traces, $\xrightarrow[]{l}$ is not a weak bisimulation.
\end{theorem}
\begin{proof}
We can prove this theorem by providing a counterexample. Here, we just give the idea of it: consider the process (with identifier $0$) that evaluates \texttt{let X = 0 in X}. This process terminates in two sequential steps according to the semantics. By the definition of the weak bisimulation, taking a reduction step determined by any action (specifically for $\arrive{0}{1}{\exit{\texttt{'kill'}}{\texttt{'false'}}}$), the result configurations should be reducible to each other (by two sequential steps). \texttt{let X = 0 in X} evaluates to a dead process with the action above, which naturally could not be reduced to anything by sequential steps.
\end{proof}

In this section we defined program equivalence based on bisimulations and proved that sequential evaluation is a bisimulation. With the help of these definitions and theorems, we can establish the equivalence of simple programs. As we noted at the start of the section, all of the definitions and theorems presented here are formalised in the Coq proof assistant~\cite{coreerlangmini}. We plan to further investigate bisimulations to enable reasoning about more complex programs.

\section{Related Work}\label{sec:work}

The results presented in this paper are extensions to our work on sequential Core Erlang~\cite{horpacsi2022program,bereczky2020machine,bereczky2021validation,bereczky2020core}. We mainly based the concurrent semantics on the work of Fredlund~\cite{fredlund2001framework}, Harrison~\cite{harrison2017coerl}, and Lanese et al.~\cite{lanese2019playing}. The general idea of an interaction semantics of actor languages is described in the work of Mason and Talcott~\cite{mason1991equivalence}.

The formalisation of Fredlund~\cite{fredlund2001framework} is the most detailed regarding both the sequential and concurrent parts of Erlang, which also faithfully follows the documentation of Erlang~\cite{erlangref}. However, it considers signal transfer as an atomic operation (i.e. when a signal is sent, it immediately arrives), while according to signal ordering guarantee~\cite{erlangrefprocs}, the order of the signals sent from an entity to the same entity is preserved, which means that two signals that are targeting different entities can arrive in arbitrary order. The semantics of Fredlund~\cite{fredlund2001framework} differentiates active and passive termination signals, while we denote these by the link flag of the exit signal.

Moreover, the work of Fredlund~\cite{fredlund2001framework} differentiates only three actions on the inter-process level semantics: input, output, and silent. This approach closely follows the general idea of the interaction semantics~\cite{mason1991equivalence}. With our approach, we can simulate input and output actions: \textit{send} actions can be considered as output actions, \textit{arr} actions are the input actions, while every other action can be regarded as silent. The advantage of our semantics is that we can distinguish more classes of reduction sequences, moreover, we also exploit this property: the theorems discussed in \Cref{sec:property} and \Cref{sec:bisim} involving sequential reductions would not hold, if other actions (e.g. $\setflag$) were considered as silent.

Lanese et al.~\cite{lanese2019playing} describe their results on bisimulations, and prove a number of system equivalences (e.g. renaming, normalisation). They also use ethers to store messages, however, their approach (deliberately) ignores the guarantee for signal ordering~\cite{erlangrefprocs}. Moreover, they do not formalise signals except messages, and used only a small subset of Core Erlang. 
Still, we incorporated some of their ideas in the formalisation of program equivalence, and plan to pursue this topic more in detail.

The work of Vidal et al.~\cite{lanese2018theory,lanese2021casual,vidal2014towards,nishida2016reversible} is also related to ours, they define multiple semantics (reduction semantics and small-step semantics) for Core Erlang to express reversible computation. The language they formalised has a similar coverage to our formalisation, they also formalised concurrent semantics with an ether and action traces, moreover, they also proved similar theorems about the properties. However, their formalisations do not include signals except messages.

Harrison~\cite{harrison2017coerl} presented a formalisation of a minimal subset of Core Erlang in his paper, which has also been formalised in Isabelle. His formalisation techniques aided us while creating a usable Coq definition of the concurrent semantics, however, his formalisation includes only a few of the language elements, and he too treated signal transfer as an atomic operation.

An important advantage of our formalisation compared to most of the existing ones is that it is also implemented in Coq in an open-source project. Most of the existing works are paper-based formalisations or the machine-checked version is no longer available to the public. Furthermore, our semantics implements the signal ordering guarantee~\cite{erlangrefprocs} more faithfully than the  other discussed approaches.

\section{Conclusion and Future Work}\label{sec:conclusion}

In this paper, we described our three-level, modular formal semantics for concurrent Core Erlang. We discussed a number of theorems about the determinism and confluence properties of the semantics, defined bisimulations to be able to reason about program equivalence, and proved that side-effect-free evaluations of a program provide equivalent programs. Finally, we compared our approach with the results of other authors. The formalisation has also been implemented as an open-source project in the Coq proof assistant~\cite{coreerlangmini}.

In the future we are planning to further extend this formalisation. Our future goals include the following points:

\begin{itemize}
    \item Investigating bisimulations in more depth, potentially by following a similar path to Lanese \emph{et al}~\cite{lanese2019playing} who defined barbed congruence, that enabled them to develop a proof technique to effectively reason about program equivalence.
    \item Proving the equivalence between more complex examples of concurrent programs equivalent, as well as investigating equivalence between sequential and concurrent algorithms.
    \item Extending the semantics to cover exceptions and other side effects (e.g. input-output) based on our previous results~\cite{bereczky2020machine}.
    \item Implementing a formalisation of the module system within this semantics.
    \item In the longer term, an extensive, usable formalisation of Erlang is our ultimate goal.
\end{itemize}

\section*{Acknowledgements}
Supported by the ÚNKP-21-3 New National Excellence Program of the Ministry for Innovation and Technology from the source of the National Research, Development and Innovation Fund. Supported by ``Application Domain Specific Highly Reliable IT Solutions'' financed under the Thematic Excellence Programme TKP2020-NKA-06 (National Challenges Subprogramme) funding scheme by the National Research, Development and Innovation Fund of Hungary.

\bibliographystyle{actaplaindoi}
\bibliography{biblio}

\end{document}